\documentclass[conference]{IEEEtran}
\IEEEoverridecommandlockouts
\usepackage{amsmath}
\usepackage{amssymb}
\usepackage{graphicx}
\usepackage{epstopdf}
\usepackage{graphics}
\usepackage{multicol}
\usepackage{amsmath,epsfig}
\usepackage{amsfonts}
\usepackage{amssymb}
\usepackage{amsthm}
\usepackage{bbm}
\usepackage{color}
\usepackage{cite}
\usepackage{float}
\usepackage{enumerate}
\usepackage{soul}
\usepackage{subcaption}

\usepackage{caption}
\newtheorem{lemma}{Lemma}

\theoremstyle{plain}

\theoremstyle{plain}

\theoremstyle{plain}

\usepackage{mathtools, cuted}

\providecommand{\lemmaname}{Lemma}
\providecommand{\propositionname}{Proposition}
\usepackage{mathtools}
\providecommand{\theoremname}{Theorem}
\providecommand{\lemmaname}{Lemma}
\providecommand{\propositionname}{Proposition}

\newcommand{\hypgeo}[2]{%
  {\vphantom{F}}_{#1}\kern-\scriptspace F_{#2}%
}

\DeclareMathOperator{\erf}{erf}
\begin{document}
\title{A Tractable Handoff-aware Rate Outage Approximation with Applications to THz-enabled Vehicular Network Optimization}
\author{Mohammad Amin Saeidi, {\em Graduate Student Member, IEEE}, Haider Shoaib, {\em Student Member, IEEE}, and \\Hina Tabassum, {\em Senior Member, IEEE} \\
 \IEEEauthorblockA{\IEEEauthorrefmark{1} Department of Electrical Engineering and Computer Science, York University, ON, Canada}%
\thanks{M.A. Saeidi, H. Shoaib and H.~Tabassum are with York University, ON, Canada  (e-mail: amin96a@yorku.ca, haider98@my.yorku.ca, hinat@yorku.ca). This research was supported by a Discovery Grant funded by the Natural Sciences and Engineering Research Council of Canada.}
}

\maketitle

\begin{abstract}
In this paper, we first develop a tractable mathematical model of the handoff (HO)-aware rate outage experienced by a typical connected and autonomous vehicle (CAV) in a given THz vehicular network. The derived model captures the impact of line-of-sight (LOS) Nakagami-m fading channels, interference, and molecular absorption effects. We first derive the statistics of the interference-plus-molecular absorption noise ratio and demonstrate that it can be approximated by Gamma distribution using Welch-Satterthwaite approximation. Then, we show that the distribution of signal-to-interference-plus-molecular absorption noise ratio (SINR) follows a generalized Beta prime distribution. Based on this, a closed-form HO-aware rate outage expression is derived. Finally, we formulate and solve a CAVs' traffic flow maximization problem to optimize the base-stations (BSs) density and speed of CAVs with collision avoidance, rate outage, and CAVs' minimum traffic flow constraint. The CAVs' traffic flow is modeled using Log-Normal distribution. Our numerical results validate the accuracy of the derived expressions using Monte-Carlo simulations and discuss useful insights related to optimal BS density and CAVs' speed as a function of crash intensity level, THz molecular absorption effects, minimum road-traffic flow and rate requirements, and maximum speed and rate outage limits.

\end{abstract}


\section{Introduction}

Connected and autonomous vehicles (CAVs) are emerging as a key technology to enable next-generation transportation infrastructure and travel behavior \cite{10008649}. To monitor the surrounding traffic environment and make effective driving decisions in real-time, CAVs rely heavily on onboard sensing and communication technology  \cite{hina2}. Nevertheless, with limited  resources, a given CAV cannot typically guarantee the response time in rapidly varying traffic conditions, resulting in low driving efficiency and even accidents. Thus, high-speed and ultra-reliable 5G/6G network infrastructure will be critical for faster  V2I communications\cite{10179155, 9770093}. 
5G/6G communications leverage transmissions in millimeter waves [mmwaves] and Terahertz [THz] frequencies to offer tremendous data rates (in the order of multi-Giga-bits-per-second). However, the transmissions are susceptible to blockages and molecular absorption in the environment. Also, with the increasing velocity of CAVs, frequent handoffs (HOs) will happen which will reduce the achievable data rate. Thus, it is crucial to characterize the HO-aware performance of a typical vehicle in a THz vehicular network. In addition, optimizing the speed of CAV to maintain a \textit{required road-traffic flow} and \textit{reliable connectivity} with \textit{collision avoidance} is of prime relevance.

Existing research works considered analyzing the HO-aware data rate performance in conventional RF networks without CAV traffic flow or collision avoidance considerations. In \cite{7827020}, Arshad et al. used HO-cost to derive useful expressions for HO-aware data rate for mobile users. In \cite{7571101}, Ibrahim et al. used a C-plane and U-plane split architecture to derive per-user HO-aware data rate. Furthermore, in \cite{6477064}, Lin et al. used stochastic geometry to formulate HO rate where BSs were randomly distributed. In \cite{hina3}, the authors considered the outage and rate analysis considering Rayleigh fading channels and traffic flow maximization in a conventional RF network.
Recently, some research works considered THz-enabled vehicular networks. In \cite{Tanvir-Mobility}, the authors proposed mobility-aware expressions for coverage probability in a two-tier RF-THz network. The authors have suggested a reinforcement learning method for a V2I network and autonomous driving rules considering both RF and THz base-stations (BSs) in \cite{10001396}.
{The study in \cite{V2V-TWCOM} developed a framework by taking into account the best locations for antennas to assess the performance of multi-hop V2V communication that operates in sub-THz bands.}


The main gap in each of the mentioned works was that they did not consider analyzing the HO-aware rate outage in a THz communication vehicular network while considering line-of-sight (LOS) Nakagami-m fading channel, interference, and  molecular absorption noise. 

In this paper, we  characterize a closed-form tractable HO-aware rate outage expression of a typical CAV in a THz network considering Nakagami-m fading channels, log-normal distribution of the spacing between CAVs, Beer's Lambert path-loss model, interference, and molecular absorption noise.  We first characterize the statistics of the interference-plus-molecular absorption noise and approximate it with Gamma distribution using \textit{Welch-Satterthwaite approximation}. Then, we show that the distribution of SINR can be characterized as a generalized Beta prime distribution. Based on this, a closed-form rate outage expression is derived. Finally, a  traffic flow maximization problem with rate outage, collision avoidance, and traffic flow constraints is formulated and solved to optimize the CAV speed and TBS density. Numerical results validate the accuracy of the derived expressions.


\section{System Model and Performance Metrics}

\begin{figure}
    \label{sys model}
    \includegraphics[width=0.45\textwidth]{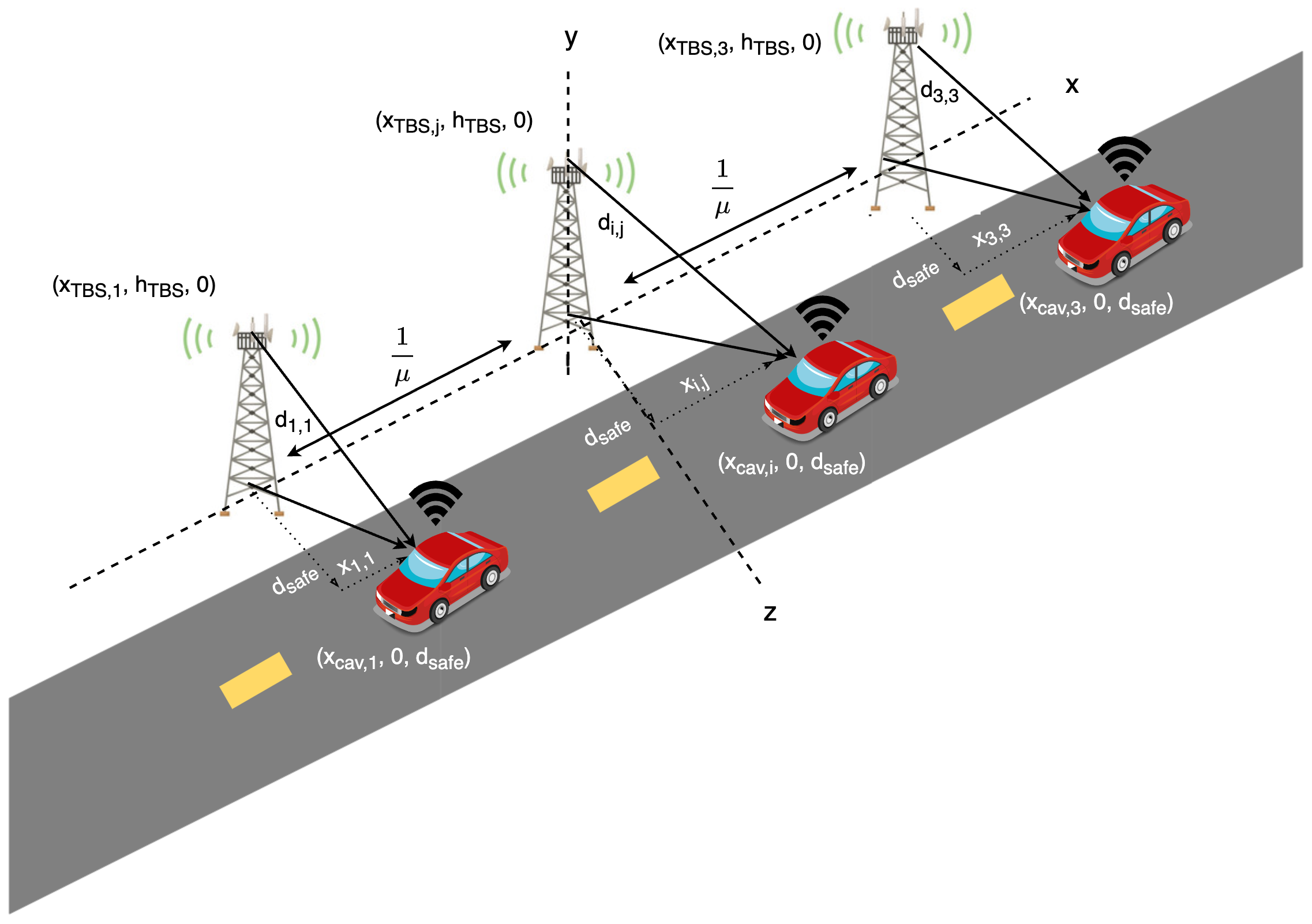}
    \centering
    \caption{V2I communication model for CAVs in THz network.}
\end{figure}

{We consider a CAV exclusive corridor equipped with $N$ TBSs alongside the corridor with a fixed length $L$ with $N_c$ CAVs of the same type. TBSs are deployed at a certain distance $d_{\mathrm{safe}}$ with a density $\mu = N/L$ TBSs per unit distance. The distance between any two TBSs is $1/\mu$ as demonstrated in Fig.~1. Furthermore, we define the density of CAVs on the road as $k$ CAVs per unit distance, and the CAVs' speed is denoted as $v$. Considering the theory of macroscopic traffic flow, we can define instantaneous traffic flow as $q=kv$ vehicles per unit time \cite{MAHNKE20051}. Furthermore, we introduce the spacing between vehicles $s$ where $k$ has an inverse relationship $k = 1/s$. Finally, with a given PDF of the density of vehicles $k$ on the corridor, i.e., $f_K(k)$, the  traffic flow can be formulated as:}
\begin{equation}
\label{flow}
    Q=\int_0^{\infty} k v f_K(k) \mathrm{d}k.
\end{equation}
The spacing between vehicles $s$ is considered to follow a log-normal distribution which is accurate for daytime hours \cite{4346442}.  The PDF of $s$ is thus given as follows:
\begin{equation}
   f_S(s) = \frac{1}{s\sigma_{\mathrm{LN}} \sqrt{2\pi}}\exp{\left(-\frac{(\ln{(s)} - \mu_{\mathrm{LN}})^2}{2\sigma_{\mathrm{LN}}^2}\right)}, 
\end{equation}
where $\mu_{\mathrm{LN}}$ and $\sigma_{\mathrm{LN}}$ are the logarithmic average and scatter parameters of log-normal distribution, respectively.

Therefore, \eqref{flow} can be rewritten as follows:
\begin{equation}
\label{avg flow}
    Q=\int_0^{\infty} \frac{1}{s} v f_S(s) \mathrm{d}s = v \exp{\left(\frac{\sigma^2_{\mathrm{LN}} - 2\mu_{\mathrm{LN}}}{2} \right)}.
\end{equation}

{We assume nearest TBS association for the connection between the CAVs and TBSs. In terms of received signal power at the $i$-th CAV, we formulate the expression assuming distance-based path-loss and short-term Nakagami fading at the transmission channel with the following expression:}
\begin{align}
\label{signal TBS}
    S_{i,j} &= P^{\text{tx}}_j G^{\textrm{tx}}G^{\textrm{rx}} \left(\frac{c}{4\pi f d_{i,j}}\right)^2  e^{-k(f) d_{i,j}} \chi_{i}\nonumber \\
    &=  P^{\text{tx}}_j \zeta d^{-2}_{i,j}e^{-k(f) d_{i,j}} \chi_{i},
\end{align}
 
In addition, we use \eqref{signal TBS} to formulate the recieved SINR at the $i$-th CAV from the $j$-th TBS with the following expression:
\begin{equation}
\label{sinr TBS}
    \mathrm{SINR}_{i,j} = \frac{S_{i,j}}{I_{i} + \phi_i + \sigma^2},
\end{equation}
where $P_j^{\mathrm{tx}}$ is the transmit power of the TBS $j$, $\zeta = G^{\textrm{tx}}G^{\textrm{rx}} \left(\frac{c}{4 \pi f}\right)^2$, $c$ and $f$ represent the speed of light and THz carrier frequency,
respectively. {$G^{\textrm{tx}}$ and $G^{\textrm{rx}}$ denote the antenna gain of TBSs and CAVs, respectively}, $d_{i,j}$ represents the distance between the $j$-th TBS and $i$-th CAV, $k(f)$ is the molecular absorption coefficient at carrier frequency $f$, and $\chi_{i, j}$ represents the fading channel of user $i$ modeled with Nakagami distribution where the power of the Nakagami fading channel follows Gamma distribution. The Nakagami model is suitable for THz transmission as it can capture communication environments with different LOS and non-LOS components using its
fading severity parameter, $m$.  The Nakagami-m distributed channel becomes Rayleigh distributed when $m$ = 1 and it can model Rician fading when $m = (K +1)2/(2K + 1)$, where $K$ is the ratio of the power in the LOS part to the power in the different multi-path elements (non-LOS). Furthermore, $\sigma^2$ is the thermal noise power at the receiver. 

Note that $I_{i} = \sum_{k, k \neq j}P^{\mathrm{tx}}_k \zeta  F d^{-2}_{i,k} \chi_{i,k}$ is the cumulative interference at the $i$-th CAV, which accounts for both interference and absorption noise from interfering TBSs. $F = F_{\mathrm{tx}} F_{\mathrm{rx}} = \frac{\theta_{\mathrm{tx}} \theta_{\mathrm{rx}}}{4 \pi^2}$ is the probability of alignment between the main lobes of the interferer and the $i$-th CAV assuming negligible side-lobe gains. 
{Finally, $\phi_i = P^{\text{tx}}_j \zeta d^{-2}_{i,j}(1-e^{-k(f) d_{i,j}}) \chi_{i,j}$ is the molecular absorption noise from TBS $j$ on CAV $i$.}

The molecular absorption coefficient of the isotopologue $t$ of gas $g$ for a molecular volumetric density at pressure $p$ and temperature $T$ can be defined as:   
\begin{equation}
    k(f) = \sum_{(t,g)} \frac{p^2 T_{\mathrm{sp}} q^{(t,g) } N_A S^{(t,g)} f \tanh{\left(\frac{h c f}{2 k_b T} \right)}}{p_0 T_k VT^2 f^{(t,g)}_c \tanh{\left(\frac{h c f^{(t,g)}_c}{2 k_b T} \right)}}F^{(t,g)}(f),
\end{equation}
where $p$ is the ambient pressure of the transmission medium, $p_0$ is the reference pressure (1 atm), $T$ is the temperature of the transmission medium, $T_{\mathrm{sp}}$ is the temperature at standard pressure, $q^{(t,g)}$ is the mixing ratio of gasses, $N_A$ is Avogadro's number, and V is the gas constant. $S^{(t,g)}$ is the line intensity which is the strength of the absorption by a specific type of
molecules, $f$ is the frequency of the EM wave, $f^{(t,g)}_{c}$ is the resonant frequency of gas $g$, $h$ is Planck's constant, and $k_b$ is the Boltzmann constant. Regarding the frequency $f$, the Van Vleck-Weisskopf asymmetric line shape is considered as follows:
\begin{equation}
    F^{(t,g)}(f) \!=\! \frac{100 c \alpha^{(t,g)}f}{\pi f^{(t,g)}_c} \!\! \left(\!\frac{1}{G^2 + \left(\alpha^{(t,g)} \right)^2} + \frac{1}{H^2 + \left(\alpha^{(t,g)} \right)^2} \! \right)\!,
\end{equation}
where $G = f + f_c^{(t,g)}$, $H = f - f_c^{(t,g)}$, and finally the Lorentz
half-width is given as follows:
\begin{equation}
    \alpha^{(t,g)} = \left(\left(1-q^{(t,g)} \right) \alpha^{(t,g)}_{\mathrm{air}} + q^{(t,g)} \alpha^{(t,g)}_0 \right) \left(\frac{p}{p_0} \right) \left(\frac{T_0}{T} \right)^\gamma,
\end{equation}
where $T_0$ is the reference temperature, the parameters air half-widths $\alpha^{(t,g)}_{\mathrm{air}}$, self-broadened half-widths $\alpha^{(t,g)}_{0}$, and temperature broadening coefficient $\gamma$ are obtained directly from the HITRAN database \cite{ROTHMAN2009533}.

{Furthermore, we note that the distance between a CAV $i$ and TBS $j$ can be calculated as 
    $d_{i,j} = \sqrt{x_{i,j}^2 + h_{\mathrm{bs}}^2 + d_{\mathrm{safe}}^2},$
where $h_{\mathrm{bs}}$ is the height of the TBSs, $d_{\mathrm{safe}}$ is the safety distance from the CAV on the road to the TBS, and $x_{i,j}$ is the distance across the $x$ axis from the TBS to the CAV's location . In addition, the traditional data rate without consideration of HOs between a TBS and CAV can be defined using Shannon's theorem as
    $R_{i,j} = W\log_2(1 + \text{SINR}_{i,j}) $,
where the bandwidth of the channel is represented by $W$.}

\section{HO-Aware Rate Outage Analysis}
In this section, we consider a THz vehicular network to first derive the HO-aware data rate, and then derive novel expressions for the PDF and CDF of the HO-aware rate outage probability and provide tractable closed-form expressions. 

To derive the HO-aware data rate, we first consider the HO-cost which is linear with respect to both HO rate (HOs per second) and HO delay (secs per HO) \cite{7571101} as 
   $ H_c = h_d\times H.$
$H$ is the handoff rate which is defined as the number of cell boundaries crossed per second. Since we are assuming nearest TBS association with equidistant TBSs, each boundary will be the same. Therefore, we formulate the HO rate as $H = \mu v$. Finally, the HO-aware data rate \cite{7827020} is formulated as: 
{
\begin{equation}
    \label{data rate final}
    M_{i,j} = R_{i,j}(1 -  H_{c, \text{max}}) = R_{i,j}(1-\min\{H_d, 1 \}),
\end{equation}
where $H_{c, \text{max}} =\min\{H_d, 1 \}$ and $H_d = h_d \mu v$.
Note that when the HO cost is greater than one, the delay caused by the handover exceeds the duration the CAV is connected to a TBS, resulting in no data transmission \cite{HO-Ref}.}
Furthermore, we introduce a minimum HO-aware data rate $R_{\text{th}}$ which will ensure a secure and consistent connection from a CAV to TBS so that each CAV attains the mandatory QoS.

The  HO-aware rate outage probability $P_{\mathrm{out}}$ is defined as:
\begin{align}
\label{outage mobile}
    P_\mathrm{out} &= \mathrm{Pr}\left(  M_{i,j} \leq R_{\mathrm{th}} \right) = \mathrm{Pr}\left(Z = \frac{S_{i,j}}{I_{i} + \phi_i + N} \leq \gamma_{\mathrm{th}} \right),
\end{align}
{where $\gamma_{\mathrm{th}}$ is the desired SINR threshold given as  $\gamma_{\mathrm{th}} = 2^{\frac{R_{\mathrm{th}}}{W(1-H_{c, \mathrm{max}})}}-1$.} 

To derive the rate outage, we initially derive the PDF and CDF of $S_{i,j}$. The random variable $S_{i,j}$ is a scaled Gamma random variable, i.e.,  $S_{i,j} = Y = a_{i,j}\chi_{i, j}$, where $a_{i,j} = P^{\text{tx}}_j \zeta d^{-2}_{i,j}e^{-k(f) d_{i,j}}$. With the use of a single variable transformation method, we derive the PDF and CDF of $Y$ as:
\begin{equation}
\label{s pdf}
    f_{Y}(y) = \frac{\beta_1^{m_j} y^{m_j-1}}{\Gamma(m_j)}e^{-\beta_1 y}, \quad
    F_{Y}(y) = \frac{\gamma\left(m_j, {y}{\beta_1} \right)}{\Gamma(m_s)},
\end{equation}
where $\beta_1 = m_j/a_{i,j}$ and $m_j$ is flexible fading parameter, {$\Gamma(.)$ and $\gamma(.,.)$ are Gamma and lower incomplete Gamma functions, respectively.}

Since the thermal noise is typically negligible compared to molecular absorption noise, we consider deriving the PDF of 
$$C_i = I_{i} + \phi_i = \sum_{ k=1, k \neq j}^{k=N} b_{i,k} \chi_{i,k} + b_{i,j} \chi_{i,j}.$$
where $b_{i,k} = P^{\mathrm{tx}}_k \zeta  F d^{-2}_{i,k}$ and $b_{i,j} = P^{\text{tx}}_j \zeta d^{-2}_{i,j}(1-e^{-k(f) d_{i,j}})$.
\begin{lemma}[Statistics of Interference-plus-Molecular Absorption Noise] 
    The interference-plus-molecular absorption noise $C_{i}$ follows a sum of $N$  scaled gamma random variables \cite{ansari2017new}, i.e., we get a weighted sum of $N$ independent but non-identical Gamma random variables with fading parameter $m_i$.  {By using the \textit{Welch-Satterthwaite approximation}   given in \cite{hina4,hina5}, the pdf of $X=C_i$ is given as follows:}
\begin{equation}
    f_X(x) = \frac{\beta_2^p x^{p-1}}{\Gamma(p)}e^{-\beta_2 x}
\end{equation}
where
\begin{align}
    \beta_2 &= \frac{m_1 \bar{b}_{i, 1} + m_2 \bar{b}_{i, 2} + \dots + m_N \bar{b}_{i, N}}{m_1 \bar{b}^2_{i, 1} + m_2 \bar{b}^2_{i, 2} + \dots + m_N \bar{b}^2_{i, N}}\\
    p &= \frac{(m_1 \bar{b}_{i, 1} + m_2 \bar{b}_{i, 2} + \dots + m_N \bar{b}_{i, N})^2}{m_1 \bar{b}^2_{i, 1} + m_2 \bar{b}^2_{i, 2} + \dots + m_N \bar{b}^2_{i, N}},
\end{align}
where $\bar{b}_{i,k} = b_{i,k}/m_k$, $\forall \ k \in \{1,...N-1\}$, and $\bar{b}_{i,j} = b_{i,j}/m_j$.
\end{lemma}

In the following lemma, we derive the outage expression.
\begin{lemma}[HO-aware Rate Outage Probability] Given the PDF  of signal $Y =S_{i,j}$ and interference-plus-molecular absorption noise $X = C_{i,j}$ of $i$th CAV, the closed-form outage expression can be given as follows:
  \begin{align}
\label{pout}
     P_\mathrm{out} &= I\left(\frac{\gamma_\mathrm{th}}{\gamma_\mathrm{th} + \frac{\beta_2}{\beta_1} }, m_j, p\right) = \frac{B \left(\frac{\gamma_{\mathrm{th}}}{\gamma_{\mathrm{th}} + \frac{\beta_2}{\beta_1} }, m_j, p\right)}{B(m_j, p)},
\end{align}  
where $I{(\cdot, \cdot, \cdot)}$ is the regularized Beta function. $B{(\cdot, \cdot, \cdot)}$ is the incomplete Beta function and  $B(\cdot, \cdot)$ is the Beta function.
\end{lemma}
\begin{proof}

Given that $X \sim \mathrm{Gamma}(p,\beta_2)$ and $Y \sim \mathrm{Gamma}(m_j,\beta_1)$,  $Z= Y/X \sim \mathrm{B}(m_j,p,1, \beta_2/\beta_1)$ follows Generalized Beta prime distribution \cite{Ratio_Of_Gamma}. The PDF of $Z$ can thus be given as follows:
\begin{equation} \label{pdf z equation}
    f_Z(z) = \frac{\beta_1\left(\frac{z}{\beta_2/\beta_1} \right)^{m_j - 1}\left(1 + \frac{z}{\beta_2/\beta_1} \right)^{-m_j - p}}{\beta_2 B(m_j, p)},
\end{equation}
The CDF of $Z$ can then be given as follows:
\begin{align}
\label{z cdf equation}
      F_Z(z) = I\left(\frac{z}{z + \frac{\beta_2}{\beta_1} }, m_j, p\right) = \frac{B \left(\frac{z}{z + \frac{\beta_2}{\beta_1} }, m_j, p\right)}{B(m_j, p)}
\end{align}
Finally, the outage in \textbf{Lemma~2} can be calculated by substituting $z=\gamma_{\mathrm{th}}$ in \eqref{z cdf equation}. Note that $\gamma_{\mathrm{th}}$ is a function of velocity.
\end{proof}

For further tractability, we derive a worst-case rate outage expression in which the worst-case interference is considered at the CAV. This happens when the CAV is halfway between any pair of TBSs (i.e. $\frac{1}{2\mu}$), i.e., $$d^{\mathrm{worst}}_{i,k} = \sqrt{h_{\mathrm{bs}}^2 + d_{\mathrm{safe}}^2 + \frac{(2 k+1)^2}{4 \mu ^2}}.$$
Subsequently, the worst-case rate outage probability can be given as follows: $$  P^{\mathrm{worst}}_{\mathrm{out}}  = \frac{B \left(\frac{\gamma_{\mathrm{th}}}{\gamma_{\mathrm{th}} + \frac{\beta_2^{\mathrm{worst}}}{\beta_1} }, m_j, p^{\mathrm{worst}}\right)}{B(m_j, p^{\mathrm{worst}})},$$
where 
$$ \beta_2^{\mathrm{worst}}= \frac{ \sum_{ k=1, k \neq j}^{k=N} \bar{b}^{\mathrm{worst}}_{i,k} m_k + \bar{b}_{i,j} m_j
}{ \sum_{ k=1, k \neq j}^{k=N} (\bar{b}_{i,k}^{\mathrm{worst}})^2 m_k + \bar{b}^2_{i,j} m_j},$$ and 
    $$ p^{\mathrm{worst}} = \beta_2^{\mathrm{worst}}\left(\sum_{ k=1, k \neq j}^{k=N} \bar{b}^{\mathrm{worst}}_{i,k} m_k + \bar{b}_{i,j} m_j\right). $$
    where $\bar{b}^{\mathrm{worst}}_{i,k} = P^{\mathrm{tx}}_k \zeta  F (d_{i,k}^{\mathrm{worst}})^{-2}/m_k$.
\section{Traffic flow Maximization with Maximum Rate Outage and Minimum Traffic Flow  Constraints}
In this section, we formulate the traffic flow maximization problem with constraints including collision avoidance, minimum traffic flow, and maximum rate outage so as to optimize CAV speed $v$ and TBS density $\mu$. We derive a closed-form solution for the optimal speed and a numerical search approach to obtain TBS density.
The optimization problem of maximizing traffic flow is stated in the following.
\begin{align}
 &(\textbf{P1}) \quad \max_{\mu, v}\qquad Q = v \exp\left(\frac{\sigma^2_{\textrm{LN}}-2\mu_{\textrm{LN}}}{2}\right) \nonumber\\
    \mathrm{s.t.} & (\textbf{C1})\quad  v\leq \frac{\exp{\left(\sigma_{\mathrm{LN}} \sqrt{2} \erf^{-1}{\left(2 \epsilon - 1 \right) + \mu_{\mathrm{LN}}} \right)}}{\tau} =V_{\mathrm{safe}}\nonumber\\&
    \textbf{(C2)} \quad v\geq  \frac{Q_{\mathrm{min}}}{\exp{\left(\frac{\sigma^2_{\mathrm{LN}} - 2\mu_{\mathrm{LN}}}{2} \right)}} = V_{\mathrm{flow}} \nonumber\\&
    (\textbf{C3}) \quad  { I\left(\!\frac{\gamma_{\mathrm{th}}(\mu,v)}{\gamma_{\mathrm{th}}(\mu,v) + \frac{\beta_2^{\mathrm{worst}}(\mu,v)}{\beta_1} }, m_j, p^{\mathrm{worst}}(\mu,v)\!\!\right) \! \leq \! O_{\mathrm{th}} } \nonumber\\&
    (\textbf{C4}) \quad 0 < v \leq V_{\text{max}} \nonumber
    \\&
    (\textbf{C5}) \quad 0 \leq \mu \leq \mu_{\text{max}}  \nonumber
\end{align}
\textbf{C1} is the collision avoidance constraint which guarantees that the speed of each CAV does not exceed the safe speed as: 
\begin{equation}
\label{avoidance}
    \mathrm{\mathrm{Pr}}\left(v\geq \frac{s}{\tau}\right) = \mathrm{\mathrm{Pr}}\left(s \leq v \tau \right) \leq \epsilon,
\end{equation}
where a probability is considered since $s$ is a random variable, $\tau$ is the processing decision time for the CAVs, and $\epsilon$ is the crash intensity level. Since the expression is in terms of a probability, we substitute the CDF of $s$ to derive a closed-form expression as follows: $\mathrm{\mathrm{Pr}}\left(v\geq \frac{s}{\tau}\right) = \frac{1}{2}\left(1 + \mathrm{erf}\left(\frac{\ln{(v\tau)}-\mu_{\mathrm{LN}}}{\sigma_{\mathrm{LN}} \sqrt{2}} \right)\right)$. Finally, to simplify the optimization problem, we isolate $v$ in the form which is written in \textbf{C1}. In terms of this constraint, it represents the safe speed that keeps the crash intensity level below $\epsilon$ where $\tau$ and $\epsilon$ can be changed to be more strict or more lenient in regards to the probability of a CAV collision. 

Furthermore, \textbf{C2} is the minimum traffic flow constraint where we rearrange \eqref{avg flow} to isolate $v$ in a similar manner as \textbf{C1}. This constraint states that the speed should be above a given minimum guaranteed traffic flow.  
Furthermore, to ensure that the rate outage does not exceed a maximum limit, we introduce \textbf{C3} in which $O_{\mathrm{th}}$ is the maximum allowed rate outage limit.
Finally, we include \textbf{C4} which cuts the CAV speed at a maximum speed limit, and \textbf{C5} which sets a limit on the maximum TBS density.

{We observe that \textbf{P1} is a non-linear programming problem because of constraint \textbf{C3} which is a non-linear function of $\mu$ and $v$. However, we note that the objective function increases monotonically with the increase in $v$. Therefore, in order to maximize the velocity, we formulate the velocity in constraint \textbf{C3}, i.e., $V_{\textrm{data}}$ as a non-linear function of $\mu$.}
 
{For a given $O_{\textrm{th}}$, $V_{\textrm{data}}$ and $\mu$ are obtained as follows:}
{
\begin{equation}
    V_{\textrm{data}}=f(\mu) = \frac{1-\frac{R_{\textrm{th}}}{W \log_2 (1+\gamma(\mu))}}{h_d \mu},
\end{equation}
}where 
$$\gamma(\mu) = \frac{\frac{\beta_2^{\mathrm{worst}}(\mu)}{\beta_1} A(\mu)}{1-A(\mu)}, \ \ A(\mu) =  I^{-1}(O_{\textrm{th}} , m_j, p(\mu)),$$
and $I^{-1}(.,.,.)$ denotes the inverse of regularized beta function. 
By using \texttt{fminbnd} solver in MATLAB, which utilizes golden-section search algorithm, we can obtain the optimal $\mu^*$ that maximizes $V_{\textrm{data}} = f(\mu^*)$ for a given $O_{\textrm{th}}$.

\begin{lemma}
Since $Q$ is increases linearly with $v$, and $v$ is restricted from $V_\mathrm{max}$, $V_\mathrm{safe}$, and $V_\mathrm{data}$, we derive the optimal speed $v^*(\mu)$ which maintains each constraint as demonstrated
\begin{equation}
 \label{Q opt prob}
v^*=
    \begin{cases}
    V_\mathrm{data} &  V_\mathrm{flow} < V_\mathrm{data} < \min\{V_\mathrm{max}, V_\mathrm{safe}\}\\
  \min\{V_\mathrm{max}, V_\mathrm{safe}\} &   V_\mathrm{flow} < \min\{V_\mathrm{max}, V_\mathrm{safe}\} < V_\mathrm{data}  \\
    \end{cases}
\end{equation}

\end{lemma}

\section{Numerical Results and Discussions}
{In this section, we validate the accuracy of derived expressions  by comparing the analysis and Monte-Carlo simulations. We also demonstrate the effect of important traffic flow and THz network parameters such as crash probability level, TBS density, data rate threshold, and molecular absorption.}

The system parameters  will be consistent from hereon unless it is stated otherwise. $V_{\mathrm{max}} = 30$ m/s, $\mu_{\text{max}}=0.5 \ m^{-1}$, $h_d = 0.35$ s/HO, $W = 3$ GHz, $G^{\mathrm{tx}} = 25$ dB, $G^{\mathrm{rx}} = 25$ dB, $c = 3\times 10^{8}$ m/s, $f = 0.837$ THz, $P^{\textrm{tx}}_j = 0.2$ W, $F = 0.0069$, $m_j = 2$, $m_i = 2.5$, $\epsilon = 2\%$, $\tau = 5\times 10^{-3}$ sec, $\mu_{\mathrm{LN}} = 0$, $\sigma_{\mathrm{LN}}=1$, $d_{\mathrm{safe}} = 5$ m, $h_{\mathrm{bs}} = 8$ m, $L = 2000$ m. 

\begin{figure}
  \centering
\includegraphics[scale=0.52]{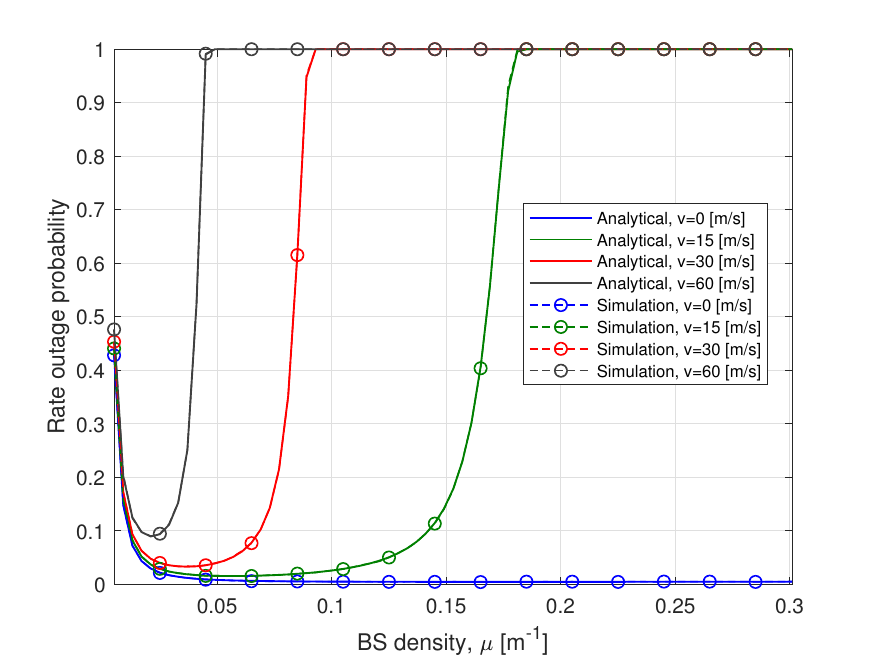}
\caption{Rate outage versus BS density, $\mu$, for different velocity. $R_{\textrm{th}}=1$ [Gbps]}
\label{fig:Outage-Vs-mu}
\end{figure}

\begin{figure}[t]
    \centering
     \subfloat[]{\includegraphics[scale=0.52]{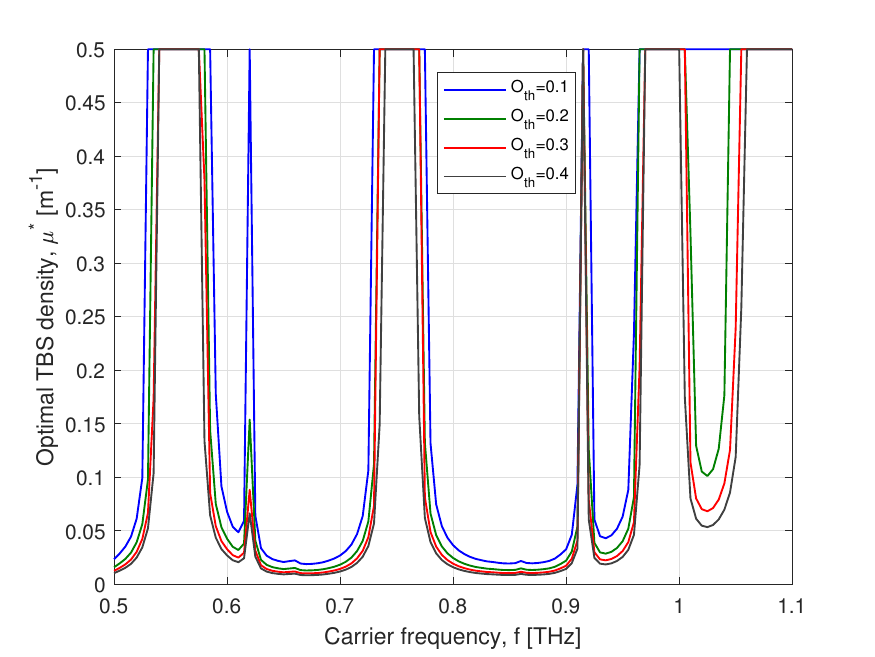}
    \label{fig:mu_opt_Vs_fTHz}}
    \vfill
    \subfloat[]{\includegraphics[scale=0.52]{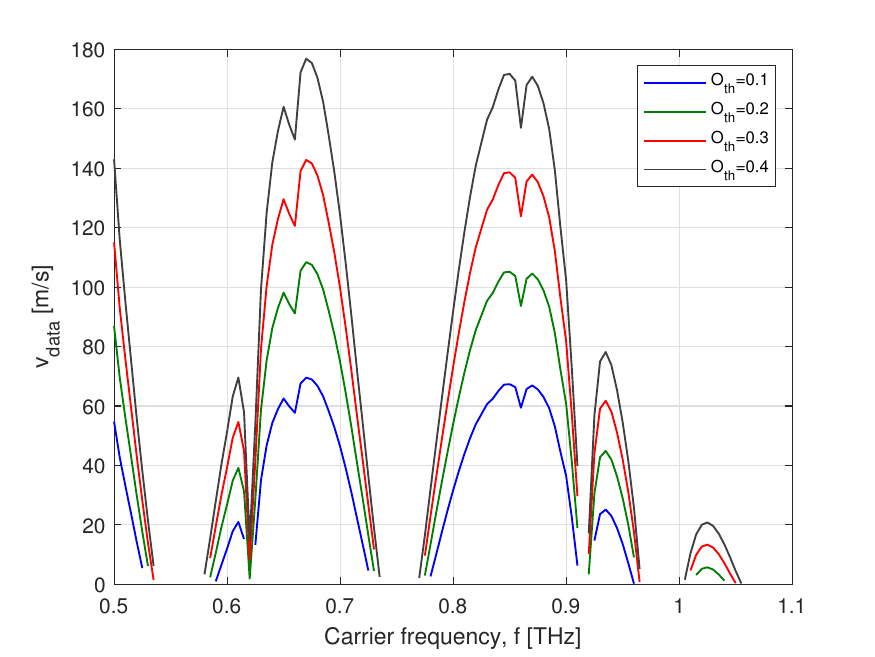}
    \label{fig:Vdata_Vs_fTHz}}
    \caption{Impact of using various carrier frequency in THz band on the optimal TBS density and $v_{\textrm{data}}$ for different rate outage threshold $O_{\textrm{th}}$. $R_{\textrm{th}}=1$ [Gbps], (a) Optimal TBS density $\mu^*$, (b) $V_{\textrm{data}}$}
    \label{fig:mu_opt_Vdata_Vs_fTHz}
\end{figure}

\begin{figure}[t]
    \centering
     \subfloat[]{\includegraphics[scale=0.52]{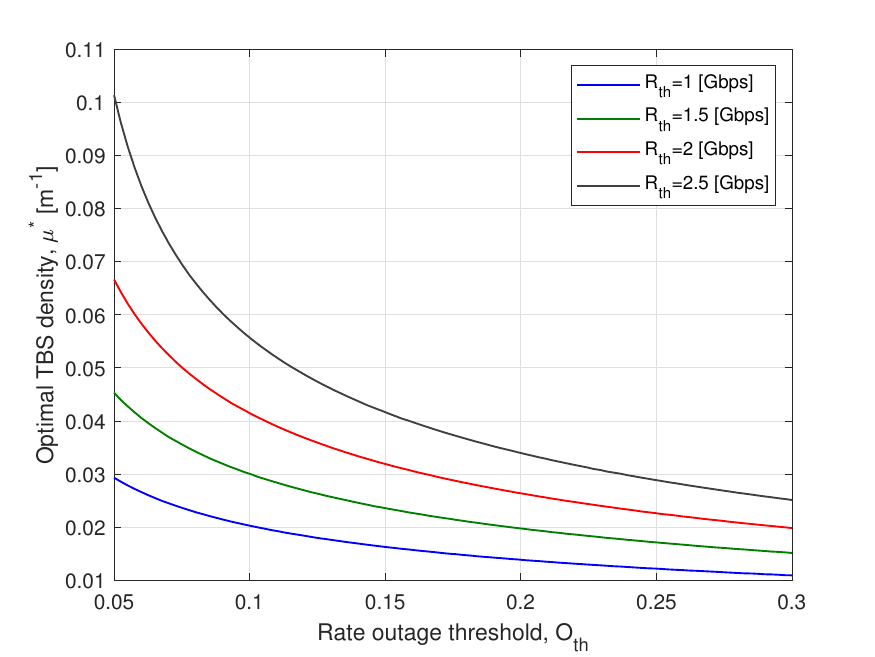}
    \label{fig:Vdata_Vs_Oth}}
    \vfill
    \subfloat[]{\includegraphics[scale=0.52]{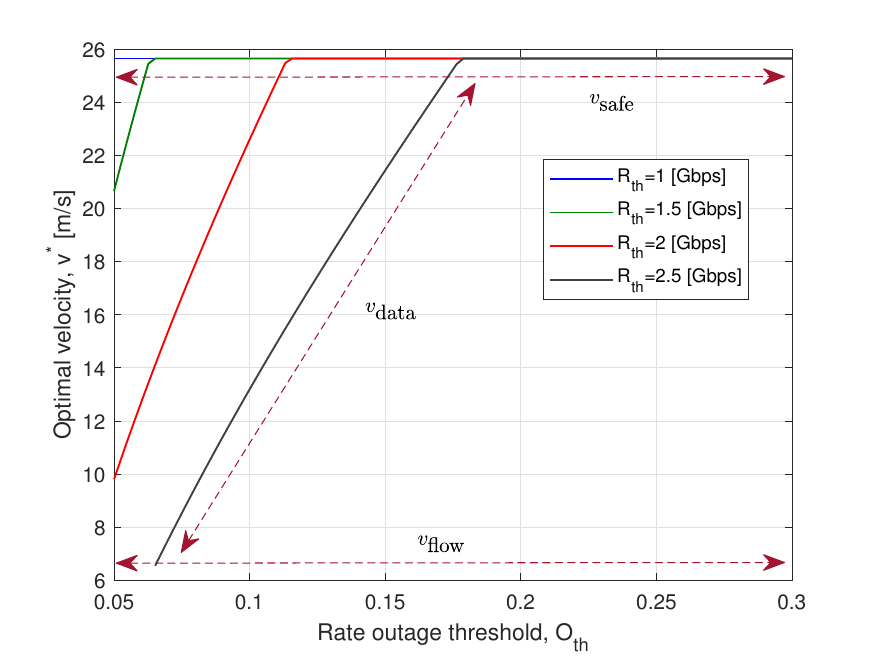}
    \label{fig:mu_opt_Vs_Oth}}
    \caption{Impact of increasing minimum rate outage threshold on the optimal TBS density and optimal velocity $v^*$ for different minimum rate requirement, $R_{\textrm{th}}$ (a) Optimal TBS density $\mu^*$, (b) Optimal velocity $v^*$}
    \label{fig:mu_opt_Vdata_opt_Vs_Oth}
\end{figure}

\begin{figure}
  \centering
\includegraphics[scale=0.52]{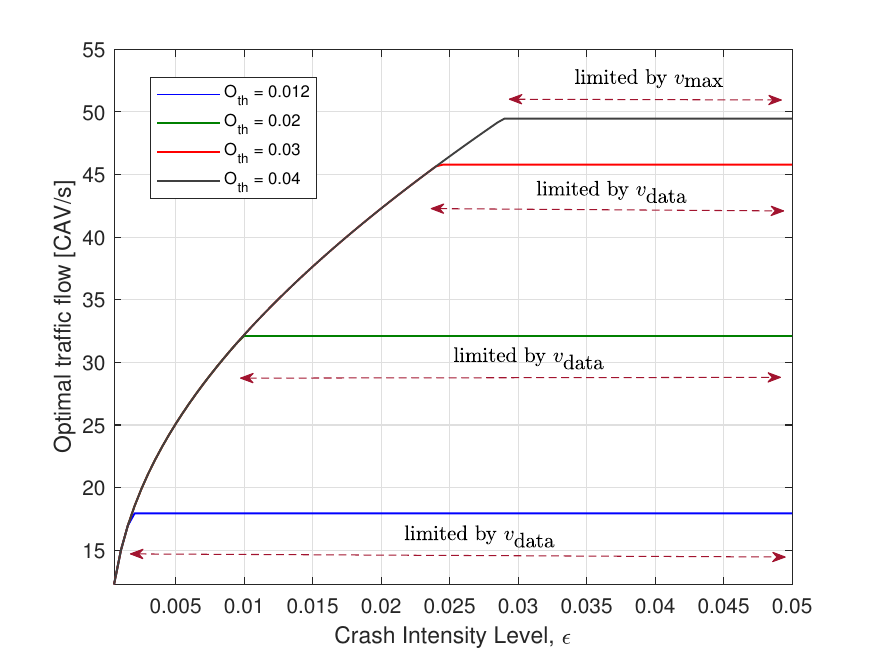}
\caption{Optimal traffic flow versus crash intensity level. $Q_{\min}=10$ [CAV/s], $R_{\textrm{th}}=1$ [Gbps]}
\label{fig:Q_traffic-Vs-epsCrash}
\end{figure}


Fig.~\ref{fig:Outage-Vs-mu} demonstrates the THz rate outage probability as a function of $\mu$ for different CAV speeds $v$. The analytical expression matches perfectly with the simulation results. We observe that when $\mu$ increases, at first the outage decreases due to improved signal strength from TBSs; however, the outage starts to increase beyond a certain point due to an increase in the overall interference and in the amount of HOs which lowers the data rate. We note that the outage probability increases with higher CAV speeds for similar reasons.

{Fig.~\ref{fig:mu_opt_Vdata_Vs_fTHz} illustrates the effect of utilizing various carrier frequencies in the THz band on the optimal TBS density and $V_{\textrm{data}}$. With the higher values of rate outage threshold, $\mu$  needs to be very high and and $V_{\textrm{data}}$ needs to be very low (such that the problem becomes infeasible most of the time) to meet the outage limits at frequencies within the regions of high molecular absorption. Moreover, Fig.~\ref{fig:Vdata_Vs_fTHz} demonstrates that by selecting appropriate transmission windows, THz transmissions can allow for much higher CAV velocities while maintaining the outage threshold and minimizing the corresponding optimal TBS density. 
}

Fig.~\ref{fig:Vdata_Vs_Oth} depicts optimal density of BSs required as a function of rate outage limit $O_{\mathrm{th}}$. As the rate outage tolerance increases, the velocity threshold of CAV $V_{\mathrm{data}}$ increases from constraint \textbf{C3}. That is, the higher rate outage threshold  will allow a CAV to move faster with a less reliable connectivity. However, this increase in speed results in more HOs; thus optimal TBS density reduces. The other observation is that if one wants to ensure a given outage threshold $O_{\mathrm{th}}$ but with improved data rates, more BSs need to be deployed, i.e, optimal TBS density increases. Fig.~\ref{fig:mu_opt_Vs_Oth} highlights the increased optimal velocity as a function of increasing rate outage limit $O_{\mathrm{th}}$. This initial trend of optimal velocity is due to increase in  $V_{\mathrm{data}}$. However, the optimal velocity is limited by $\min{(V_{\mathrm{max}}, V_{\mathrm{safe})}}$. This demonstrates the solution provided in \eqref{Q opt prob}

 Fig.~\ref{fig:Q_traffic-Vs-epsCrash} shows the optimal traffic flow as a function of the crash  tolerance level. The higher  crash tolerance level, the higher velocity  can be attained for CAVs due to increase in $V_{\mathrm{safe}}$ in constraint \textbf{C1} (i.e., of course at the expense of increased crashes). Furthermore, to achieve high connection reliability (i.e., reduced rate outage probability), the traffic flow will need to be sacrificed as can be noted. However, a minimum traffic flow is always guaranteed which is $Q_{\mathrm{min}} = 10$ CAVs per second.


\section{Conclusion}
{This paper presents a tractable mathematical model of HO-aware rate outage for a THz vehicular network. First, we derive the statistics of the SINR by deriving novel PDF and CDF expressions. Next, we formulate a traffic flow maximization problem considering traffic flow constraints and an outage probability constraint to obtain optimal TBS density and CAV speed. We derive a closed-form optimal CAV speed equation and numerical solution for TBS density. Finally, the numerical results verify the validity and accuracy of the derived expressions and extract insights regarding TBS density, rate outage probability, and the optimal CAV speed.}
\bibliographystyle{IEEEtran}
\bibliography{references}

\vfill
\end{document}